\documentclass[a4paper, 11pt]{article}

\usepackage{fullpage}
\usepackage[english]{babel}
\usepackage[latin1]{inputenc}
\usepackage{mathrsfs}
\usepackage{amsfonts}
\usepackage{amssymb}
\usepackage{amsmath}
\usepackage{amsthm}
\usepackage{graphicx}
\usepackage[usenames,dvipsnames]{color}
\usepackage[colorlinks=true,citecolor=blue,linkcolor=magenta]{hyperref}
\usepackage{lmodern}
\usepackage{microtype}
\usepackage{braket}
\usepackage{tabularx,colortbl}
\usepackage{pifont}
\usepackage{mathtools}
\usepackage{color}
\usepackage{dsfont}
\usepackage{algpascal}
\usepackage{algorithm} 
\usepackage{algpseudocode}
\usepackage{chngcntr}
\usepackage{complexity}
\usepackage{float} 
\usepackage[table, x11names]{xcolor}
\usepackage{epstopdf}
\usepackage{braket}
\usepackage{caption}   
\usepackage{subfig}   
\usepackage[titletoc,title]{appendix}

\newtheorem{theorem}{Theorem}
\newtheorem{corollary}{Corollary}

\newtheorem{definition}{Definition}
\newtheorem{lemma}{Lemma}

\newcommand{\comment}[1]{}

\newcommand{\idty}[1]{\mathbb{1}} 
\newcommand{\ovsqrt}[1]{\frac{1}{\sqrt{2}}}

\newcommand{\tr}{\operatorname*{Tr}{ }}

\newcommand{\V}{{\mathcal{V}}}

\newcommand{\Dd}{{\mathcal{D}}}

\newcommand{\Mm}{{\mathcal{M}}}

\newcommand{\Cc}{{\mathcal{C}}}
\newcommand{\Hy}{{\mathcal{S}}}
\newcommand{\Tr}{\mathrm{Tr}}
\newcommand{\Pp}{{\mathcal{P}}}

\usepackage[backend=biber, natbib=true, bibstyle=alphabetic,citestyle=alphabetic, doi=true, sorting=nyt, giveninits]{biblatex}
\title{A condition under which classical simulability implies efficient state learnability}

\begin{document}

\author{Mithuna Yoganathan\thanks{Email: \href{mailto:my332@cam.ac.uk}{my332@cam.ac.uk}}\\
University of Cambridge}

\date{}
\maketitle

\begin{abstract}
In the task of quantum state learning, one receives some data about measurements performed on a state, and using that, must make predictions on the outcomes of unseen measurements. Computing a prediction is generally hard but it has been shown that learning can be performed efficiently for states that are generated by Clifford circuits, which are known to be  efficiently classically simulable. This naturally leads to the question, how does efficient state learnability compare with efficient classical simulation? In this work we introduce an extra condition on top of classical simulablity that guarantees efficiently learnability. To illustrate this we prove two new examples of efficient learnability: states with low (Schmidt rank) entanglement and states described by an `efficient' ontological model.
\end{abstract}
 
\section{Introduction}

Finding a quantum state in agreement with a set of expected values of observables is a common task in quantum information processing. If the set of observables is complete, i.e. forms a basis in the space of density matrices, then the problem reduces to quantum state tomography and it is known to require exponentially many measurements~\citep{haah2017sample}.

The problem of PAC learning quantum states, which has been formally introduced in~\citep{aaronson2007learnability}, does not aim do full tomography of a quantum state. Instead the goal is to produce a hypothesis state that is likely to agree with the true quantum state when an unseen measurement is preformed on it. This ability is know as the \textit{generalisation} capability of a learning algorithm. Aaronson proved that, under suitably defined notions of learning, a set of measurements with a number of elements that grows linearly with the number of qubits contains sufficient information to construct an hypothesis that generalises well. The apparent contradiction with the exponential bounds of quantum state tomography can be resolved by noting that in the definition of learnability adopted by Aaronson the hypothesis state need not be `close' to the true state at all, but only needs to have similar expectation values for randomly selected measurements.

A second notable feature of Aaronson's learning theorems is that, although from an information theoretic perspective the problem can be solved efficiently, the underlying computational task required to generate the hypothesis is generally computationally hard. This raises the question of whether learning quantum states can not only be informationally but also computationally efficient. In fact this is not possible in general but may be true for restricted sets of states and measurements. Indeed there is one known example, stabiliser states under Pauli measurements, where learning becomes computationally tractable~\citep{rocchetto2017stabiliser}. 

Stabiliser states are well known to be efficiently classically simulable via the Gottesman--Knill algorithm \cite{gottesman1997stabilizer}. Many of the tools of that algorithm are used in the hypothesis learning algorithm, suggesting a link between simulation and learnability. However, not all classically simulable states are efficiently learnable. What addition features must a classical simulation have in order to be efficiently learnable? In this paper we identify an extra condition, called efficient invertibility, that allows a classical simulation to be extended into an efficient learning algorithm. 

An efficient classical simulation demands that the states are efficiently describable in a form that can be easily updated by unitaries and for which measurement statistics can be calculated. Informally, efficient invertibility requires that given an outcome for a particular measurement it must be possible to (usefully) describe all states that could have given that outcome efficiently. For example, if a stabiliser state gives outcome $+1$ with certainty for a $Z_1$ measurement, then the stabiliser state must have $Z_1$ in its stabiliser. 

We show two new examples of efficient learnability and show they satisfy this efficient invertibility condition. The first is low Schmidt rank states under measurements that are not too entangling. This was shown to be simulable by Vidal \cite{vidal2003efficient}. In the second example we consider states that are efficiently described by a hidden variable or ontological model, and show that such states are efficiently simulable and learnable.

\subsubsection*{Organisation}

We introduce notations and important concepts in Section~\ref{sec:prelims}. In Section~\ref{sec:invert} we discuss why classical simulation on its own does not imply efficient learnability, and introduce the idea of efficient invertibility. In Section~\ref{sec:entangled} and Section~\ref{sec:EOM} we present the two examples.

\section{Classical simulation and efficient learnability}
\label{sec:prelims}

In this work we only consider $2$ outcome POVMs with outcomes $0$ and $1$. We identify a measurement with the first POVM element $E$. The expected value of the measurement is therefore given by $\Tr (E \rho)$.

We focus on classes of quantum computation with an $n$ qubit input $\rho$, acted on by a unitary $U$, but in which only one (prespecified) line is measured in the $Z$ basis. It is possible to think of this computation as a measurement of $\rho$ by the POVM $E=U^{\dagger} Z U$. This follows from the observation that $U$ followed by the $Z$ measurement has the same effect as performing $U^{\dagger} Z U$ followed by $U$. However, as we are only interested in the measurement outcomes, if we commute $U$ to after the measurement, we no longer need to perform $U$. This way of characterising computations as measurements will be useful to us when comparing simulability and learning. We therefore define efficient simulability in this picture.

\begin{definition} \label{def:classsim}
\textit{(Efficient classical simulation)} Let $\Cc_n$ be a class of $n$-qubit quantum states, and $\Cc=\bigcup_n \Cc_n$. Let $\Mm_n$ be a set of two-outcome measurements on $n$ qubits, and $\Mm=\bigcup_n \Mm_n$. We say $\Cc$ is efficiently classically simulable with respect to $\Mm$ if:
\begin{itemize}
\item there exists a canonical classical description of each $E\in\Mm_n$ of length at most $\poly(n)$;
\item there exists a canonical classical description of each $\rho\in\Cc_n$ of length at most $\poly(n)$;
\item there exists a classical algorithm $L$, that takes as input the above descriptions and outputs $\Tr(E\rho)$ to additive error $\epsilon$;
\item $L$ has runtime polynomial $n$ and $1/\epsilon$.  
\end{itemize}
\end{definition}

This notion of simulability is equivalent to weak simulation. A weak simulation of a quantum circuit is one that outputs samples that are similar to samples from the true quantum computation; they come from a distribution close to the true distribution in $l_1$ (additive) distance. The ability to weakly simulate enables one to compute $\Tr(E\rho)$ to $1/\poly(n)$ additive error (by repeating the computation and taking the expectation) and vice versa. However, if instead of $1/\poly(n)$ we required that $\Tr(E\rho)$ is computed to $1/\exp(n)$ additive error, this definition would be that of strong simulation. It is so called, because this simulation is able to compute $\Tr(E\rho)$ to an accuracy beyond what the original computer is capable of. Weak simulation is a more natural definition of simulation, and so we use it here. However we note that several of the simulations mentioned in this paper are actually also simulations in the stronger sense.

When learning a quantum state in the PAC sense, the goal is to find an hypothesis state that describes the training data well. Let $\rho$ be an unknown quantum state and $\sigma$ the hypothesis state. First, we want the ability to find, efficiently, a description of a quantum state $\sigma$ consistent with given measurement outcomes. If the outcome of the measurement $E_i$ is known to be $\Tr(E_i \rho)$, we require that $|\Tr(E_i  \sigma) - \Tr(E_i \rho) |$ is small. This ensures that our hypothesis state would have acted similarly to the true state $\rho$ on previously seen data. Mathematically, finding such a state corresponds to solving a semidefinite programming (SDP) feasibility problem (recall that the \textit{feasible region} of an optimisation problem is the set of points that satisfies the problem's constraints). For an introduction to the theory of SDP we recommend~\citep{boyd2004convex}. Second, we require the ability to compute efficiently new predictions. This second requirement can be interpreted as performing a classical simulation of the hypothesis state. We now formalise the above discussion. 
We begin by giving a formal statement of the learning theorem proved by Aaronson.
\begin{theorem} \label{thm:PAC}
\textit{(Quantum Occam's Razor \citep{aaronson2007learnability})} As above, let $\Cc_n$ be a class of $n$-qubit quantum states, and $\Cc=\bigcup_n \Cc_n$. Let $\Mm_n$ be a set of two-outcome measurements on $n$ qubits, and $\Mm=\bigcup_n \Mm_n$. Let $\{E_1,...,E_m\}$ be $m$ measurements drawn from the distribution $\mathcal{D}_n$ over $\Mm_n$. For a given $\rho\in \Cc_n$ and call $T=\{(E_i,\Tr(E_i \rho))\}_{i\in[m]}$ a training set. Suppose there is a hypothesis state $\sigma$ such that $|\Tr(E_i \sigma)-\Tr(E_i \rho)|\leq \gamma\epsilon/7$ for each $i\in[m]$. In words, the hypothesis state is consistent with the data. Then, with probability at least $1-\delta$, the hypothesis state will also be consistent with new data:
\begin{equation}
\Pr_{E\sim D_n}[|\Tr(E\sigma)-\Tr(E\rho)|\leq \epsilon]\geq 1-\delta,
\end{equation}
provided there was enough training data:
\begin{equation} \label{eq:datasize}
m\geq \frac{C}{\sigma^2 \epsilon^2}\Big(\frac{n}{\gamma^2\epsilon^2}\textrm{log}^2\frac{1}{\gamma\epsilon}+\textrm{log}\frac{1}{\delta}\Big).
\end{equation}
\end{theorem}

This theorem shows that every set of quantum states $\Cc$ is \textit{PAC learnable} with respect to measurements $\Mm$ (distributed according to $\mathcal{D}$) with only a linear in $n$ number of data points. To produce a good hypothesis state with high probability, one only needs to find a state that fits the current data. However, finding such a state is not guaranteed to be computationally efficient.

For some particular classes of states and measurements, learning a hypothesis state may be efficient. The definition of efficiently learnable we use in this work is as follows. 
\begin{definition} 
\label{def:effLearn}
\textit{(Efficiently learnable)} 
Let $\Cc$ and $\Mm$ be as above, and assume there is some canonical efficient classical description of every $E\in\Mm_n$, so $E$ can be described using $\textrm{O}(\poly(n))$ bits \footnote{It is natural to assume one exists, because otherwise it would take exponential time for the experimenter to even describe the measurement they will perform.}. We say $\Cc$ is efficiently learnable with respect to $\Mm$ if there exists a pair of classical algorithms $L_1$ and $L_2$ with the following properties:
\begin{itemize}
\item ($L_1$, SDP feasibility algorithm)  For every $\rho \in \Cc_n$, $\eta>0$, and for every training set $T=\{(E_i,\Tr(E_i \rho))\}_{i\in[m]}$ , $L_1$ outputs a $\poly(n)$ bit classical description of a hypothesis state $\sigma$ that satisfies the following approximate feasibility problem
\begin{gather}\label{eq:learningSDP}
\nonumber |\Tr(E_i  \sigma) - \Tr(E_i \rho) | \leq \eta \quad \textrm{for all} \enspace i \in [m], \\
\sigma \succeq 0, \\ 
\nonumber \Tr(\sigma) = 1.  
\end{gather}
$L_1$ runs in time $\poly (n,m,1/\eta)$.
\item ($L_2$, simulation algorithm) For every $n$ qubit $\sigma$ which is an output of $L_1$ and for every $E\in \Mm_n$, $L_2$ computes $\Tr(E \sigma)$. $L_2$ runs in time $\poly (n)$.
\end{itemize}

\end{definition}

We define $\mathcal{S}$ to be the set of all possible hypothesis states $\sigma$ that satisfy Eq.~\ref{eq:learningSDP}. In the language of learning theory $\mathcal{S}$ is the hypothesis state space of the quantum state learning problem. Note that the definition above ensures that all hypothesis states have efficient classical descriptions.

Previous to this work, there has only been one example of efficient learnability: stabiliser states with respect to Pauli measurements \citep{rocchetto2017stabiliser}.
We remark that our notion of efficiently learnability is slightly different from previous definitions~\citep{aaronson2007learnability,rocchetto2017stabiliser}. In our definition we emphasise that efficient learnability does not only require an efficient way to produce a description of the hypothesis state $\sigma$, but also that this description can be used to efficiently compute $\Tr(\sigma E)$ for all $E\in \Mm$. In the language of classical simulation that was introduced at the beginning of the section, this amounts to efficiently simulating $\mathcal{S}$ under the measurement set $\Mm$. The motivation for this extra requirement is that efficient learnability ought to imply that, given the training data, a classical computer can efficiently produce a prediction for a new measurement of the state.

SDPs such as the one in Eq.~\ref{eq:learningSDP} are well studied and can be solved efficiently in the dimension of the input matrices using interior point methods~\citep{nesterov1994interior, alizadeh1995interior}, the ellipsoid method~\citep{grotschel2012geometric} or algorithms specifically designed for the problem of learning quantum states~\citep{hazan2008sparse}. Because for quantum states the dimensions of the input matrices scale exponentially with the system size, in general, quantum states are not efficiently learnable.

\section{Efficient invertibility}
\label{sec:invert}

 Classical simulabilty gives some, but not all, of the conditions needed for learnability. If an efficient classical simulation algorithm exists for $\Cc$ with respect to $\Mm$, that guarantees that an efficient description exists for all $\rho\in\Cc$ and that, given the efficient description, it is possible to compute $\Tr(\rho E)$ for all $E\in\Mm$. Recall that these properties are required in the definition of efficient learnability (in the case where $\Hy=\Cc$). This is called $L_2$ or the simulation algorithm in Def. \ref{def:effLearn}. However, this does not guarantee that the SDP feasibility algorithm (also required in the definition) exists. In fact, Aaronson provides a counterexample \citep{aaronson2007learnability}. If all classically simulable circuits are (even quantumly) learnable, that implies there are no cryptographic one-way functions (safe from a quantum attack). 

In this section we provide an extra property on top of simulability that implies learnability. To illustrate this, we examine the efficient learning algorithm in the example of stabiliser states under Pauli measurements.

The algorithm follows this procedure (elaborated in \citep{rocchetto2017stabiliser}):

\begin{enumerate}
\item For each training set data point $\{E_i, \Tr(E_i\rho)\}$, let $\Hy_i$ be the set of stabiliser states consistent with this measurement outcome: $\Hy_i=\{\phi\in \Cc: \Tr(E\phi)=\Tr(E\rho)\}$. It is possible to characterise states in this set in terms of their stabilisers.

\item It is also possible to characterise states in the intersection $\bigcap \mathcal{S}_i$ in terms of their stabilisers. States in this set are consistent with all data points. Choose one such state $\sigma$.

\item Use the Gottesman-Knill Theorem to make predictions using $\sigma$ for future Pauli measurements. 
\end{enumerate}

While the classical simulability of stabiliser states under Pauli measurements was necessary for the final step, we see that there is another property that is important: the ability to efficiently characterise all states consistent with the data. In other words, to describe $\Hy_i$ efficiently, and from this, efficiently compute the intersection of the sets. Then finally, it must be possible to efficiently describe a particular state $\sigma$ in this set. In particular, that efficient description must be the one that the classical simulation algorithm uses. In this case, the final description of $\sigma$ is in terms of its stabilisers, which allows us to use the Gottesman-Knill simulation.

We call this condition the \textit{efficient invertibility condition}, as we require to invert $\tr(E_i \rho)$ (i.e. to find the set of states that matches the expectation for each training measurement $E_i$). Any state in the intersection corresponds to a solution to the SDP feasibility problem. 

Classical simulation is related to, but not surficient for, efficient invertibility. If $\Cc$ is classically simulable with respect to $\Mm$, that guarantees that an efficient description exists for all $\rho\in\Cc$, and that given the efficient description, it is possible to compute $\Tr(\rho E)$.
A trivial learning algorithm based on classical simulation would then take $\Hy = \Cc$, compute the expectation value of all the elements in the hypothesis set and return any of the states that is consistent with all the measurements. This does `invert' $\tr(E_i\rho)$. However, in general $|\Cc|= |\Hy| = O(\textrm{exp}(n))$ and thus this algorithm is not an efficient inversion. 
In the stabiliser example it was important that, aside from $\Tr(\sigma E)$ being efficient to calculate in those cases, it is was also efficient to characterise all states $\phi$ that had a particular value of $\Tr(\phi E)$. In the next subsection we provide another example where the efficient invertibility condition is crucial. 

We now generalise the algorithm developed for stabiliser states, making allowance for the hypothesis state to only be consistent with the data to additive error $\eta$ (as per definition \ref{def:effLearn}). This learning algorithm can be schematised as follows:
\begin{enumerate}
\item characterise all states in $\mathcal{S}^{\eta'}_i=\{\phi\in\Cc: |\Tr(E_i \rho)-\Tr(E_i \phi)|\leq \eta'\}$ (with $0<\eta'\leq \eta$) in time polynomial in $n$ and $1/\eta$ and, 
\item produces an efficient description of a state $\sigma$ in the intersection of all $\mathcal{S}^{\eta'}_i$, $i\in[0,m]$, in time polynomial in $m$, $n$ and $1/\eta$, 
\item and then uses this description to compute $\Tr(\sigma E)$, for any $E\in\Mm$, in time polynomial in $n$.
\end{enumerate}
If all the conditions hold then $\Cc$ is efficiently learnable with respect to $\Mm$. 
This may not be the only condition under which efficient learning is possible. However, it appears to be natural as it holds for the three known examples of efficient learning, namely stabilisers and the examples we provide in the following sections.

\section{Slightly entangled states are efficiently learnable}
\label{sec:entangled}

Quantum computations in which the state of the computer does not become entangled are classically simulatable \citep{jozsa2003role,vidal2003efficient}. In fact, even if the states are allowed a `slight' amount of entanglement, as defined below, it is still possible to efficiently simulate their evolution. Here we show that such states are also efficiently learnable under suitable measurements. 

\begin{definition}
\textit{(Schmidt decomposition)} The Schmidt decomposition of a state $|\psi\rangle$ with respect to the $A|B$ partition is
\begin{equation}
|\psi\rangle_{AB}= \sum_{i=1}^r \sqrt{\lambda_i}|a_i\rangle_A |b_i\rangle_B, 
\end{equation}
where the $|a_i\rangle$ are orthonormal, as are the $|b_i\rangle$. Assume $\lambda_i\neq0$ for all $i$. \textit{The Schmidt rank} of the $A|B$ partition is $r$. 
\end{definition}

See \citep{nielsen2002quantum} for the proof that a Schmidt decomposition exists for all states and all partitions, and that the Schimdt rank is unique. 

Note that if the Schmidt rank is $1$, the $A$ and $B$ subsystems are not entangled. Generally though the Schmidt rank can be exponentially large in the dimension of the subsystems. This motivates the following definition. 

\begin{definition}
\textit{(Slightly entangled states)} A class of states is called \textit{slightly entangled} if, for every $n$ qubit state $|\psi\rangle$ in the class and every partition $A|B$, the Schmidt rank is bounded by a polynomial. This class is called \textit{$L$-entangled} if the polynomial bounding the Schmidt rank is $L(n)$.
\end{definition}

\citep{vidal2003efficient} shows that if the state of a quantum computation is $L$-entangled throughout the computation, this computation can be classically simulated. \citet{yoran2006short} give an example of circuits that have this property. More specifically, they show that a computation which is $\log(n)$ depth, has bounded range interactions, and has product inputs has a state during the computation that is L-entangled and hence classically simulable. Ref \citep{jozsa2006simulation} generalised this to show the following. Suppose $\C$ is a circuit comprised on $1$ and $2$ qubit gates. Let $D_i$ be the number of gates that either act of line $i$, or act across that line (i.e $D_i$ is the number of two qubit gates that act on lines $j$ and $k$ with $j\leq i\leq k$). Let $D=\max_i D_i$. Then if $D$ is bounded by a polynomial and the input to the circuit $\C$ is a product state, then the quantum state during the computation is slightly entangled, and hence $\C$ is efficiently classically simulable. 

We now define low Schmidt rank measurements as the measurements induced by circuits in the above form. 
\begin{definition}
\textit{(Low Schmidt rank measurements)} Suppose $U$ is a $n$ qubit unitary comprised of $1$ and $2$ qubit gates, with $D=\max_i D_i$ as above, and suppose $D$ is bounded by a polynomial in $n$. Then we call the measurement $U^\dagger Z_i U$ a \textit{low Schmidt rank measurement}, where $Z_i$ is a $Z$ measurement of the $i$-th line. 
\end{definition}

This definition is motivated by noting that if such a measurement where performed on a slightly entangled state, the result would be classically simulable in the sense of~\ref{def:classsim}. The name refers to the property that such a measurement keeps the Schmidt rank of a slightly entangled state low. It does not imply that these measurements are low rank. 

\begin{theorem}
The class of $L$-entangled states is efficiently learnable under low Schmidt rank measurements. 
\end{theorem}

\begin{proof}
We first give an efficient parametric description of an unknown $L$-entangled state. This description involves $2nL^2$ complex parameters at most and any state in that form is necessarily an $L$-entangled state. 
The method for deriving this description is described in detail on page $5$ in~\citep{jozsa2006simulation} where a proof of its correctness is also given. We omit that proof here, but explain the process. 

Suppose $|\psi\rangle$ is a $L$-entangled $n$ qubit state, and consider the partition of the qubits $1|2...n$. Use the Schmidt decomposition to write 
\begin{equation}
|\psi\rangle=\sum_j |a_j\rangle|\nu_j\rangle,
\end{equation}
where the Schmidt coefficients are absorbed into $|a_j\rangle$, so they are orthogonal and subnormalised while $|\nu_j\rangle$ are orthonormal. Let $|a_j\rangle=\alpha_j^0|0\rangle+\alpha_j^1|1\rangle$. For $|\psi\rangle$ to be properly normalised, $\sum_j |\alpha_j^0|^2+|\alpha_j^1|^2=1$. 

Now consider the partition $12|3...n$ of $|\psi\rangle$. Let $|\nu_k\rangle$ be the orthonormal Schmidt vectors on qubits $3,...,n$. Then
\begin{equation}
|\nu_j\rangle=\sum_k |b_{jk}\rangle|\eta_k\rangle.
\end{equation}
The $|b_{jk}\rangle$s are subnormalised. Let $|b_{jk}\rangle=\beta_{jk}^0|0\rangle+\beta_{jk}^1|1\rangle$. Then because $|\nu_j\rangle$ is normalised and all $|\eta_k\rangle$ are orthonormal, $\sum_j |\beta_{jk}^0|^2+|\beta_{jk}^1|^2=1$, for all $j$. 

Continuing in this way, we can write all $L$-entangled states parametrically as 
\begin{equation} \label{eq:form}
|\psi\rangle = \sum_{j,k,l..p} |a_j\rangle |b_{jk}\rangle|c_{kl}\rangle...|k_p\rangle,
\end{equation}
with normalisation equations as above. Each index sums at most to $L$, and hence there are at most $2nL^2$ complex parameters in this description. 

Now we use this description to show such states are learnable under low Schmidt rank measurements. Consider the training set $\{(E_i,\tr(|\psi\rangle\langle\psi| E_i))\}_{i\in[m]}$. The steps of the learning algorithm follow the same pattern as the previous examples: 
\begin{enumerate} 
\item Let $E_i$ be the projector onto the positive subspace of the measurement $U^\dagger Z_j U$. Then it is possible to find the $L$-entangled states consistent with the data point $(E_i, \tr(|\psi\rangle\langle\psi| E_i))$ by first evolving the state in Eq. \ref{eq:form} by $U$. This is possible to do efficiently by the method given in~\citep{vidal2003efficient}. Then it is also possible to efficiently find the expression for the probability of outcome $+1$ when line $j$ of this state is measured in the $Z$ basis. This expression is a polynomial in the parameters. Equating that with $ \tr(|\psi\rangle\langle\psi| E_i)$ gives a polynomial equation that $|\psi\rangle$ needs to satisfy in order to be consistent with the data point. 
\item To find a $|\sigma\rangle$ consistent with all the data, all the above polynomial equations must be solved simultaneously. The system can be undetermined but we do not have to find the solution corresponding to $|\psi\rangle$. Any solution yields a description of a hypothesis state in the form of Eq.~\ref{eq:form}.
\item Given this description, it is efficient to compute the measurement outcomes of $|\sigma\rangle$ for any low Schmidt rank measurement. 
\end{enumerate}

\end{proof}

\section{Efficient ontological models are efficiently simulable and learnable}
\label{sec:EOM}

In this section we show that if a set of states can be described efficiently in the ontological models framework, then they are both efficiently simulable and efficiently learnable. The ontological models framework was introduced in~\citep{harrigan2010einstein} and is a general framework for theories that reproduce the statistics of quantum mechanics.

\begin{definition}
\textit{(Ontological model)} 
Let $\Cc$ be a set of quantum states, and $\Mm$ be a set of measurements previously\footnote{We assume that a full context is prespecified for each of the measurements in $\Mm$. This allows us to talk about ontological models that are contextual without ambiguity.}. Then an ontological model with underlying ontic state space $\Lambda_n$ describes the states in $\Cc_n$ with respect to $\Mm_n$, if the following conditions hold. There exists a function $f: \Lambda_n \times \Mm_n \rightarrow [0,1]$; for all $\rho\in\Cc_n$, and there exists a probability distribution $p_{\rho}$ over $\Lambda_n$ such that, for all $\rho\in \Cc_n$ and $E \in \Mm_n$, $\Tr(\rho E)=\sum_{\lambda\in\Lambda_n} p_{\rho} f(\lambda, E)$.

\end{definition}
The ontic state space represents the true (potentially hidden) state of a system in the quantum state $\rho$. In this model, the system is thought to be in some ontic state $\lambda\in\Lambda$ with probability $p_\rho(\lambda_n)$. There is a function that assigns the probability of outcome $+1$ for measurement $E$ as $f(\lambda,E)$. For a deterministic hidden variable theory, $f(\lambda,E)$ is always $0$ or $1$. To recover the usual quantum formalism instead, let $\Lambda_n$ be the $n$ qubit quantum state space and $p_\rho$ be a delta Dirac function probability centred on $\rho$. Finally, let $f$ assign Born Rule probabilities. Hence, this model is rich enough to describe the usual quantum formalism as well as hidden variable descriptions of quantum mechanics. 

We now add a very natural extra restriction that will make the ontological models we consider classically simulable. 

\begin{definition}
\textit{(Efficient ontological model (EOM))} An efficient ontological model for $\Cc$ under $\Mm$ is an ontological model for which 
\begin{itemize}
\item $|\Lambda_n|=O(\poly(n))$,
\item there exists a classical algorithm for evaluating $f(\lambda,E)$ and $p_\rho(\lambda) $, for all $\rho\in \Cc_n$, $\lambda\in\Lambda_n$, and $E\in \Mm_n$, in runtime polynomial in $n$.
\end{itemize}
\end{definition}

\begin{theorem}
If there exists an EOM for the set of states $\Cc$ and measurements $\Mm$, then it is possible to classically simulate $\Cc$ with respect to $\Mm$ in the sense of Definition \ref{def:classsim}.
\end{theorem}

\begin{proof}
Let $\rho$ a state in $\Cc$ and $E$ be a measurement in $\Mm$. The probability distribution $p_{\rho}(\lambda)$ can be computed for each $\lambda\in\Lambda$. $|\Lambda|=O(\textrm{poly}(n))$ and so this can be done for every $\lambda$ efficiently, and hence $p_{\rho}(\lambda)$ can be sampled from efficiently. To efficiently compute $\Tr(\rho E)$, (1) sample $\lambda'$ from $p_{\rho}$, (2) compute $f(\lambda',E)$, (3) repeat this algorithm polynomially many times to estimate $\Tr(\rho E)$ to sufficient accuracy. 
\end{proof}

So far we have only considered probability distributions $p_\rho$ over $\Lambda$ that represent the state of a quantum system $\rho$, and therefore produce the outcome statistics of $\rho$. However, it is possible to consider more general probability distributions over $\Lambda$, which are what we'll call \textit{preparations}. The outcome statistics for these will generally not correspond to any state or obey the Born rule. However, they will be useful to consider for us. This is because if the states and measurements being learned are described by an EOM, instead of producing a hypothesis quantum state $\sigma$ and using that to predict future outcome statistics, it is simpler to produce a hypothesis preparation instead. The outcomes predicted from such a preparation will in fact be a good prediction with high probability, in the usual learning sense. We now formalise this discussion.

\begin{definition} (Preparation) 
For an ontological model with ontic state space $\Lambda$, a preparation is a probabilistic mixture over $\Lambda$. Suppose that the probability distribution is given by $p: \Lambda\rightarrow [0,1]$. Identify the preparation with $p$. 
\end{definition}

The PAC learning theorem for quantum states (Theorem ~\ref{thm:PAC}) does not apply in this instance, as not all preparations are representations of quantum states. We require a generalisation of the PAC theorem in this case in order to justify the following procedure as `learning'. We provide such a generalisation:

\begin{theorem} \label{thm:EOMPAC}
\textit{(Occam's razor for EOMs)} EOMs are PAC learnable, in the sense of Theorem ~\ref{thm:PAC}. 
\end{theorem}
The proof of this theorem is in Appendix~\ref{appendix}.

\begin{theorem}
EOMs are efficiently learnable. 
\end{theorem}

\begin{proof}
Suppose that there is an EOM with ontic state space $\Lambda$ ($|\Lambda|=O(n)$) with measurement outcome assigning function $f: \Lambda \times \Mm \rightarrow \{0,1\}$. Let $p$ be the unknown preparation to learn and suppose the training data is $T=\{(E_i,d_i=\sum_{\lambda} p(\lambda) f(\lambda, E_i))\}_{i\in[m]}$. Then let the hypothesis state space be all possible preparations on $\Lambda$. 

\begin{enumerate}
\item The preparations (exactly) compatible with the data point $(E_i,d_i=\sum_{\lambda} p(\lambda) f(\lambda, E_i))$ are $\{q: d_i=\sum_{\lambda} q(\lambda) f(\lambda, E_i))\}$. It is efficient to write this condition in full, as the number of terms in the equation is polynomial in $n$ and computing $f(\lambda, E_i)$ is also efficient, by assumption. Also add the constraint that $\sum_{\lambda} q(\lambda) =1$.
\item Each of these $m$ constraint equations are linear, and involve $|\Lambda|=O(n)$ unknowns. It is possible to find a solution in time polynomial in $m$ and $n$. There may not be a unique solution, but any solution (for which q is a probability distribution) is sufficient. 
\item Given the hypothesis preparation $q$, it is efficient to compute $\sum_{\lambda} q(\lambda) f(\lambda, E)$ for any $E\in\Mm$. 
\end{enumerate}
\end{proof}

This theorem shows that preparations are efficiently learnable in the sense of Definition ~\ref{def:effLearn}.

\section*{Acknowledgements}

We thank Andrea Rocchetto for his innumerable contributions to this work. We would also like to thank Scott Aaronson for his helpful comments. 
\printbibliography

\appendix

\section*{Appendix} \label{appendix}

A set of functions is PAC learnable as long the functions are not too `flexible'. This is what allows one to predict future values of the function from a small amount of training data. The flexibility of a set of functions is quantified by the fat-shattering dimension. 

\begin{definition} (Fat-shattering dimension)
Let $\Mm$ be a sample space, and $\V$ be a set of functions that map elements of $\Mm$ to numbers in $[0,1]$. We say a set $\{E_1,...,E_k\}\subset \Mm$ is $\gamma$-fat shattered by $\V$ if there exists real numbers $\alpha_1,...,\alpha_k$ such that all $B\subseteq\{1,...k\}$, there exists $g\in\V$ such that for all $i\in\{1,...,k\}$,
\begin{itemize}
\item if $i\notin B$ then $g(E_i)\leq \alpha_i-\gamma$, and
\item if $i\in B$ then $g(E_i)\geq \alpha_i+\gamma$.
\end{itemize}
The $\gamma$-fat-shattering dimension of $\V$, $\textrm{fat}_\V(\gamma)$, is the maximum $k$ such that some $\{E_1,...,E_k\}$ is $\gamma$-fat-shattered by $\V$. 
\end{definition}

In the quantum case, let $\Cc$ be a class of quantum states, and $\Mm$ be a set of measurements. The function $\Tr(\cdot \rho): \Mm\rightarrow [0,1]$ takes measurements in $\Mm$ and outputs the probability of getting outcome $+1$ when measuring $\rho$. $\V$ is the set of these functions, for all $\rho\in\Cc$. 

For the preparations of an EOM case, if $\Pp$ is the set of preparations, $\V$ is the set of functions in the form $\sum_{\lambda} p(\lambda) f(\lambda, \cdot)): \Mm\rightarrow [0,1]$, for $p\in \Pp$. 

We now state a result that links the learnability of a class of functions with its fat-shattering dimension.

\begin{theorem} \label{classical} 
\textit{(From~\citep{anthony1995function})}
Define $\Mm$ and $\V$ as above, and let $\Dd$ be a probability distribution over $\Mm$. Fix an element $g\in\V$, as well as error parameters $\epsilon, \eta,\gamma>0$, with $\gamma>\eta$. Suppose $\{E_1,...,E_m\}$ are drawn from $\Pp$ independently according to $\D$. Let $T=\{(E_i,g(E_i))\}_i^m$ be the training set. Let $h\in\V$ be a hypothesis consistent with the training set: $|h(E_i)-g(E_i)|\leq \eta$ for $i\in[1,m]$. Then there exists positive constant K such that if 
\begin{equation} \label{mbound}
m\geq \frac{K}{\epsilon} \Big( \textrm{fat}_\V \big(\frac{\gamma-\eta}{8}\Big)\textrm{log}^2\big(\frac{fat_\V(\gamma-\eta/8)}{(\gamma-\eta)\epsilon}\big)+\textrm{log}\frac{1}{\delta}\Big),
\end{equation}
then with probability at least $1-\delta$, 
\begin{equation}
\textrm{Pr}_{E\sim \D} [|h(E)-g(E)|>\gamma]\leq \epsilon.
\end{equation}

\end{theorem}

Bounding the fat-shattering dimension for preparations of EOMs allow us to bound the expression in Equation \ref{mbound} in this case. First we prove a result about random access codes using preparations of an EOM. The equivalent theorem for the quantum and classical state was proved in~\citep{ambainis2002dense}.

\begin{theorem} \label{randomaccess}
Let $k$ and $n$ be positive integers with $k>n$. For a $k$-bit string $y=y_1...y_k$, let $p_y$ be a preparation for an EOM. Suppose there exists measurements in $\Mm$ such that, for all $y\in \{0,1\}^k$ and $i\in \{1,...,k\}$,
\begin{itemize}
\item if $y_i=0$ then the probability of outcome $+1$ for measurement $E_i$ is greater than $p$,
\item if $y_i=1$ then the probability of outcome $-1$ for measurement $E_i$ is greater than $p$.
\end{itemize}
Then $O(\textrm{log}(n))\geq (1-H(p))k$, where $H$ is the binary entropy function. 
\end{theorem}
This theorem bounds the size of strings that can be encoded in a preparation. The measurement $E_i$ returns the correct value of $y_i$ with probability $p$. 

\begin{proof}
Let $Y=Y_1...Y_k$ be the random chosen uniformly at random from $\{0,1\}^k$. For a preparation $p_y$ with $y\sim Y$, let $X$ be a random variable that chooses an ontic state $\lambda\in\Lambda$ with probability $p_y$. Let $Z_i$ be the random variable that records the result of measuring a preparation $p_{y\sim Y}$ with $E_i$, and let $Z=Z_i...Z_m$.

The mutual information of $Y$ and $X$ is bounded by the number of bits needed to specify $\lambda\sim X$. 
\begin{equation}
I(Y:X)\leq S(X)\leq \textrm{log}(|\Lambda|)=O(\textrm{log}(n))
\end{equation}
We now bound the mutual information in the other direction. 
\begin{equation}
I(Y:X)\leq S(Y)-S(Y|X)=k-S(Y|X),
\end{equation}
and, 
\begin{equation}
S(Y|X)\leq S(Y|Z)\leq\sum_{i=1}^k S(Y_i |Z)\leq \sum_{i=1}^k S(Y_i|Z_i).
\end{equation}
Because knowing $z_i\sim Z_i$ gives us the correct value of $y_i\sim Y$ with probability greater than or equal to $p$, $S(Y_i|Z_i)\leq H(p)$. Hence, $O(\textrm{log}(n))\leq k(1-H(p))$.

\end{proof}

This theorem is necessary to prove the fat-shattering theorem, which is analogous to Theorem 2.6 in~\citep{aaronson2007learnability}.

\begin{lemma}
Let $k$, $n$ and $\{p_y\}$ be as in Theorem \ref{randomaccess}. Suppose there exists $E_1,...,E_k\in \Mm$ and real numbers $\alpha_i,...,\alpha_k$, such that for all $y\in{0,1}^k$ and $i\in\{1,...,k\}$,
\begin{itemize}
\item if $y_i=0$ then the probability of outcome $+1$ for measurement $E_i$ is greater than $\alpha_i-\gamma$,
\item if $y_i=1$ then the probability of outcome $+1$ for measurement $E_i$ is less than $\alpha_i+\gamma$.
\end{itemize}
Then $O(\textrm{log}(n))/\gamma^2=O(k)$

\end{lemma}

The proof is essentially the same as for Theorem 2.6 in~\citep{aaronson2007learnability}. Interpreting $k$ as the fat-shattering dimension gives us our result. 

\begin{corollary}
For all $\gamma>0$, the fat-shattering dimension of preparations of an EOM is $O(n/\gamma^2)$.
\end{corollary}

Combining this with Theorem \ref{classical} shows that preparations of an EOM are PAC learnable.

\end{document}